\documentclass[submission,copyright,creativecommons]{eptcs}
\usepackage{amssymb,amsmath,mathpartir,paralist,amsthm,xspace}
\usepackage[utf8]{inputenc}
\usepackage[textsize=small]{todonotes}
\usepackage[T1]{fontenc}

\newcommand{\Nat}{\mathbb{N}}
\newcommand{\lY}{$\lambda Y$}
\newcommand{\arr}{\mathbin{\to}}
\newcommand{\otyp}{\mathsf{o}}
\newcommand{\ord}{\mathit{ord}}
\newcommand{\inc}{a}
\newcommand{\zero}{e}
\newcommand{\nd}{\mathit{b}}
\newcommand{\redb}{\to_\beta}
\newcommand{\pr}{{\mathsf{pr}}}
\newcommand{\np}{{\mathsf{np}}}
\renewcommand{\r}{\otyp}%{\mathbf r}}
\newcommand{\dom}{\mathit{dom}}
\newcommand{\low}{\mathit{val}}
\newcommand{\restr}{{\restriction}}

\newcommand{\Pp}{\mathcal{P}}
\newcommand{\Tt}{\mathcal{T}}
\newcommand{\Ttrip}{\mathcal{\widehat T}}%\!\!T}}
\newcommand{\AppR}{\TirName{(\!@\!)}\xspace}
\newcommand{\LamR}{\TirName{($\lambda$)}\xspace}
\newcommand{\VarR}{\TirName{(Var)}\xspace}
\newcommand{\lamdots}{.\cdots{}.}
\newcommand{\unbound}{\mathsf U}

\newtheorem{theorem}{Theorem}
\newtheorem{lemma}[theorem]{Lemma}
\newtheorem{corollary}[theorem]{Corollary}
\theoremstyle{definition}
\newtheorem{example}{Example}
\newtheorem{openpr}{Open Problem}

\providecommand{\event}{ITRS 2018}
\title{Intersection Types for Unboundedness Problems%
	\thanks{Work supported by the National Science Centre, Poland (grant no.\ 2016/22/E/ST6/00041).}}
\author{Paweł Parys
	\institute{Institute of Informatics, University of Warsaw, Poland}
	\email{parys@mimuw.edu.pl}
}
\def\titlerunning{Intersection Types for Unboundedness Problems}
\def\authorrunning{P. Parys}
\begin{document}
\maketitle

\begin{abstract}
	Intersection types have been originally developed as an extension of simple types, but they can also be used for refining simple types. 
	In this survey we concentrate on the latter option; more precisely, on the use of intersection types for describing quantitative properties of simply typed lambda-terms.
	We present two type systems.
	The first allows to estimate (by appropriately defined value of a derivation) the number of appearances of a fixed constant $a$ in the beta-normal form of a considered lambda-term.
	The second type system is more complicated, and allows to estimate the maximal number of appearances of the constant $a$ on a single branch.
\end{abstract}

\section{Introduction}
	Intersection types have been originally developed as an extension of simple types, but they can also be used for refining simple types. 
	In this survey we concentrate on the latter option; more precisely, on the use of intersection types for describing quantitative properties of simply typed lambda-terms.
	
	We consider lambda-terms as generators of trees.
	To this end, we assume a unique ground sort\footnote{%
		Following the convention in this area, we use the word ``sort'' for simple types, and the word ``type'' for intersection types refining them.}
        $\otyp$ describing trees, and we assume some uninterpreted constants, which are functions of order at most $1$.
	Then, a beta-normal form of a closed lambda-term of the ground sort does not contain any lambda-binders---it is just an applicative term composed of the uninterpreted constants,
	and thus can be seen as a tree.
        In other words, in the effect of calling a function $a$ with some trees as arguments, we obtain a new tree with a root labeled by $a$, and with the arguments attached in the children of the root.
	
	Suppose now that we have a closed lambda-term $M$ of the ground sort, and we want to estimate some quantities concerning its beta-normal form $T$.
	As a first example of such a quantity we can take the number of appearances of some fixed constant $a$ in $T$.
	How can we read this number from the original lambda-term $M$?
	As a first approach, we can look at the number of appearances of the constant $a$ in $M$.
	This can be completely inappropriate, though, for two reasons.
	First, we can have in $M$ some appearances of the constant $a$ that will be removed during beta-reductions.
	Second, maybe the constant $a$ appears in $M$ only once, but it will be replicated a lot of times during beta-reductions.
	In order to take into account these two phenomena we design an appropriate type system;
	a type derivation for the lambda-term $M$ identifies the places in $M$ that are really responsible for producing some constants $a$ in the beta-normal form $T$,
	so that these places can be counted.
	The type system realizing this goal is presented in Section~\ref{sec:deterministic}.
	
	Another quantity of the tree $T$ is the largest number of appearances of some fixed constant $a$ on a single branch in $T$.
	While the quantity of the first kind can be called deterministic, this one is slightly more complicated, and can be called nondeterministic.
	The justification of such a name is that while looking locally at some fragment of $T$ we do not know whether the constants $a$ appearing in this fragment should be counted or not
	(i.e., whether they are located on the branch of $T$ containing the largest number of constants $a$).
	We thus have to non-locally (nondeterministically) choose some branch of $T$ on which the constants $a$ should be counted.
	A type system that allows to estimate the above quantity is presented in Section~\ref{sec:nondeterministic}.
	
	The following quantity is even more involved: what is the largest number $n$ such that 
	the binary tree with all nodes labeled by $a$ and all branches of length $n$ embeds homeomorphically in the considered tree $T$?
	In a sense, this quantity combines three elements: taking the maximum, taking the minimum, and counting.
	Indeed, we take here the maximum over all embeddings of trees with all nodes labeled by $a$ of the minimum of lengths of paths in the chosen tree
	(and internally, we count the number of constants $a$ on the chosen path).
	Unfortunately, the presented methods do not allow to estimate this quantity;
	it is an open problem to construct a type system concerning this quantity.
	
	One may wonder why we want to have the aforementioned type systems, instead of just expanding the lambda-term $M$ into its beta-normal form $T$, and computing the quantity there.
	The answer is: compositionality.
	Suppose that $M$ is an application $K\,L$.
	If we know types derivable for $K$ and for $L$, we can determine types derivable for $K\,L$.
	Moreover, knowing the quantities assigned to type derivations for $K$ and for $L$ we can determine the quantity assigned to type derivations for $K\,L$.
	We thus have a composable abstraction of every lambda-term: a set of its types, and a tuple of numbers 
	(with a bound on the size of this set and on the length of this tuple that depends only on the sort of the lambda-term, not on its size).
	Existence of such an abstraction has some interesting implications.
	
	In particular, the research presented here is motivated by applications in the area of higher-order recursion schemes.
	Recursion schemes, or equivalently terms of the \lY-calculus, form an extension of the simply typed lambda-calculus by a fixed-point operator $Y$~\cite{Damm82,KNU-hopda,Ong-hoschemes,kobayashiOng2009type}.
	Trees generated by recursion schemes can be used to faithfully represent the control flow of programs in languages with higher-order functions~\cite{Kobayashi-types}.
	We remark that the same class of trees can be generated by collapsible pushdown systems~\cite{collapsible} 
	and ordered tree-pushdown systems~\cite{DBLP:conf/fsttcs/ClementePSW15}.
	
	Intersection type systems were intensively used in the context of recursion schemes, for several purposes like
	model-checking~\cite{Kobayashi-types,kobayashiOng2009type,DBLP:conf/csl/BroadbentK13,DBLP:conf/popl/RamsayNO14},
	pumping~\cite{koba-pumping,koba-pumping-new},
	transformations of HORSes~\cite{context-sensitive-2,word2tree,downward-closure}, etc.
	Interestingly, constructions very similar to intersection types were used also on the side of collapsible pushdown systems; 
	they were alternating stack automata~\cite{saturation}, and types of stacks~\cite{ho-new,Kar-Par-pumping}.
	The type systems are also closely connected to linear logic~\cite{linear-logic-1,linear-logic-2}.

	The type system of Section~\ref{sec:deterministic} is based on the type system from Parys~\cite{numbers-journal}.
	A similar type system was used to prove that some trees generated by recursion schemes cannot be generated by so-called safe recursion schemes~\cite{ho-new}.
	The type system of Section~\ref{sec:nondeterministic} comes from Parys~\cite{itrs,diagonal-types}.
	It implies decidability of the model-checking problem for trees generated by recursion schemes against formulae of the WMSO+U logic~\cite{wmsou-schemes}.
	It also allows to solve the simultaneous unboundedness problem (aka.~diagonal problem) for recursion schemes, which was first solved in a different way~\cite{downward-closure}.

\section{Preliminaries}

	The set of \emph{sorts} is constructed from a unique basic sort $\otyp$ using a binary operation $\arr$.
	Thus $\otyp$ is a sort and if $\alpha,\beta$ are sorts, so is $(\alpha\arr\beta)$.
	The \emph{order} of a sort is defined by: $\ord(\otyp)=0$, and $\ord(\alpha\arr\beta)=\max(1+\ord(\alpha),\ord(\beta))$;
	in other words, $\ord(\alpha_1\arr\dots\arr\alpha_k\arr\otyp)=1+\max_{i\in\{1,\dots,k\}}\ord(\alpha_i)$ whenever $k\geq 1$.

	%A \emph{ranked alphabet} is a set of letters $\Sigma$, together with a function $\rank\colon\Sigma\to\Nat$ assigning ranks to letters.
	
	A \emph{signature} is a set of constants, that is, symbols with associated sorts.
	For simplicity, 
	in this paper we use a signature consisting of three constants: $\inc$ of sort $\otyp\arr\otyp$, and $\nd$ of sort $\otyp\arr\otyp\arr\otyp$, and $\zero$ of sort $\otyp$
	(it is easy to generalize the methods to an arbitrary signature, assuming that sorts of constants are of order at most $1$).

	The set of \emph{(simply typed) lambda-terms} is defined by induction as follows:
	\begin{compactitem}
	\item	constants (node constructors)---a constant of sort $\alpha$ is a lambda-term of sort $\alpha$;
	\item	variables---for each sort $\alpha$ there is a countable set of variables $x^\alpha,y^\alpha,\dots$ that are also lambda-terms of sort $\alpha$;
	\item	lambda-binders---if $K$ is a lambda-term of sort $\beta$ and $x^\alpha$ a variable of sort $\alpha$ then $\lambda x^\alpha.K$ is a lambda-term of sort $\alpha\arr\beta$;
	\item	applications---if $K$ is a lambda-term of sort $\alpha\arr\beta$ and $L$ is a lambda-term of sort $\alpha$ then $K\,L$ is a lambda-term of sort $\beta$.
	\end{compactitem}
	As usual, we identify lambda-terms up to alpha-conversion (renaming of bound variables).
	We often omit the sort annotation of variables, but please keep in mind that every variable is implicitly sorted.
	A term is called \emph{closed} when it does not have free variables.
	The \emph{order} of a lambda-term $M$, denoted $\ord(M)$, is defined as the order of the sort of $M$,
	while the \emph{complexity} of $M$ is defined as the maximum of orders of subterms of $M$.

	A sort $\alpha_1\arr\dots\arr\alpha_k\arr\otyp$ is \emph{homogeneous} if $\ord(\alpha_1)\geq\dots\geq\ord(\alpha_k)$ and all $\alpha_1,\dots,\alpha_k$ are homogeneous (defined by induction).
	A lambda-term is homogeneous if all its subterms have homogeneous sorts.
	In order to avoid some technicalities, in this paper we only consider homogeneous lambda-terms.
	This is without loss of generality, since there is a simple syntactic transformation converting every closed lambda-term of sort $\otyp$ into a homogeneous lambda-term
	having the same beta-normal form~\cite{homogeneity}.
%	The advantage of the homogenity assumption is that for every lambda-term of the form $\lambda x.K$ we have $\ord(x)=\ord(\lambda x.K)-1$.

	We use the usual notion of beta-reduction: we have $M\redb N$ if $N$ can be obtained from $M$ by replacing some of its subterms of the form $(\lambda x.K)\,L$ by $K[L/x]$.
	We recall that simply typed lambda-calculus has the properties of strong normalization and confluence, that is,
	every sequence of beta-reductions from a lambda-term $M$ eventually terminates in a unique lambda-term $N$ such that no more beta-reductions can be performed from $N$;
	the lambda-term $N$ is called the \emph{beta-normal form} of $M$.
	Observe that the beta-normal form of a closed lambda-term of sort $\otyp$ is an applicative term build of constants (it does not contain variables nor lambda-binders), 
	and thus can be seen as a tree (generated by the lambda-term).
	
	In this paper we are interested in two particular reduction strategies (i.e., strategies of choosing a redex that should be reduced next).
	In the \emph{OI strategy}, we always reduce an \emph{outermost redex}, that is, a redex that is not located inside another redex.
	Notice that if $M$ is closed and of sort $\otyp$, then every outermost redex in $M$ is also closed.
	A redex $(\lambda x.K)\,L$ is a \emph{redex of order $m$} if $\ord(\lambda x.K)=m$.
	Assuming that the lambda-term is homogeneous, we have $\ord(\lambda x.K)=m$ if and only if $\ord(x)=m-1$.
	In the \emph{RMF strategy} we always reduce a \emph{rightmost redex of the maximal order}, that is, 
	a redex $(\lambda x.K)\,L$ of some order $m$ such that in the lambda-term there is no redex of a higher order, and in $L$ there are no redexes of order $m$,\label{pag:rmf}
	and the redex is not located inside $K'$ for some order-$m$ redex $(\lambda x'.K')\,L'$.
	In other words, whenever we see an order-$m$ redex $(\lambda x.K)\,L$, we first reduce all order-$m$ redexes in $L$, then the redex itself, and then we continue reducing the resulting lambda-term.
	We also write RMF$(m)$ to make it explicit that the order of the considered redex is $m$.
	When a closed lambda-term $M$ of sort $\otyp$ has complexity $m$ (and is not in the beta-normal form), then an RMF$(m)$ reduction always exist;
	thus following the RMF strategy we first reduce all redexes of order $m$ (until reaching a term of complexity $m-1$), then all redexes of order $m-1$, and so on.
	Moreover, for an RMF$(m)$ redex $(\lambda x.K)\,L$ in such a lambda-term, all variables appearing in $L$ are of order at most $m-2$.	

	Suppose that we have two functions $f,g\colon X\to\Nat$, over some domain $X$.
	We want to define when $f$ estimates $g$.
	To this end, we say that $f$ is \emph{dominated} by $g$, written $f\preceq g$, if there exists a function $\eta\colon\Nat\to\Nat$ such that $f(x)\leq\eta(g(x))$ for all $x\in X$,
	and we say that $f$ \emph{estimates} $g$, written $f\approx g$, if $f\preceq g$ and $g\preceq f$.
	It is easy to see that $f$ estimates $g$ if and only if on every subset $Y$ of the domain $X$, the functions $f$ and $g$ are either both bounded or both unbounded.
	The above relation between functions is widely used in the area of regular cost functions (see, e.g., Colcombet~\cite{regular-cost-functions}).

	One may also consider infinite lambda-terms.
	Clearly they do not reduce to a normal form in a finite number of steps, but we can consider the (unique) normal form reached in the limit, called the \emph{B\"ohm tree}.
	As in the finite case, the B\"ohm tree of a closed lambda-term of sort $\otyp$ is a (potentially infinite) tree build out of constants.
	A \emph{recursion scheme} is a finite description of a regular (i.e., having finitely many different subterms) infinite lambda-term.

\section{Deterministic Quantities}\label{sec:deterministic}

	In this section we present a type system that allows to estimate the number of appearances of the constant $\inc$ in the beta-normal form of a lambda-term.
	The type system should be such that a type derivation for a closed lambda-term $M$ of sort $\otyp$ 
	identifies the places in $M$ that are responsible for producing some $a$-labeled nodes in the beta-normal form $T$ of $M$.
	To this end, we extend the notion of sorts by a \emph{productivity flag}, which can be $\pr$ (standing for productive) and $\np$ (standing for nonproductive).

	It may happen that a single lambda-term $K$ has multiple types; for example, $\lambda y.y\,(a\,e)$ is productive when the function (substituted for) $y$ uses its argument, and nonproductive otherwise.
	Because of that, we need intersection types (i.e., the ability of assigning multiple types to the same lambda-term).
	
	In effect, our types differ from sorts in that on the left side of $\arr$, instead of a single type, we have a set of pairs $(f,\tau)$,
	where $\tau$ is a type, and $f$ is a flag from $\{\pr,\np\}$.
	The unique atomic type is denoted $\r$.
	More precisely, for each sort $\alpha$ we define the set $\Tt^\alpha$ of types of sort $\alpha$ as follows:
	\begin{align*}
		\Tt^\otyp=\{\r\},\qquad 
		\Tt^{\alpha\arr\beta}=\Pp(\{\pr,\np\}\times\Tt^\alpha)\times\Tt^\beta,
	\end{align*}
	where $\Pp$ denotes the powerset.
	A type $(T,\tau)\in\Tt^{\alpha\arr\beta}$ is denoted as $\bigwedge T\arr\tau$,
	or $\bigwedge_{i\in I}(f_i,\tau_i)\arr\tau$ when $T=\{(f_i,\tau_i)\mid i\in I\}$.
	The empty intersection is denoted by $\top$.
	To a lambda-term of sort $\alpha$ we assign not only a type $\tau\in\Tt^\alpha$, but also a flag $f\in\{\pr,\np\}$ (which together form a pair $(f,\tau)$).
	
	Intuitively, a lambda-term has type $\bigwedge T\arr\tau$ when it can return $\tau$, while taking an argument for which we can derive all pairs (of a flag and a type) from $T$;
	simultaneously, while having such a type, the lambda-term is obligated to use its arguments in all ways described by type pairs from $T$.
	And, we assign the flag $\pr$ (productive), when this term (while being a subterm of a closed term of sort $\otyp$) increases the number of constants $a$ in the resulting tree.
	To be more precise, a term is productive in two cases.
	First, when it uses the constant $\inc$.
	Notice however that this $\inc$ has to be really used: there exist terms which syntactically contain $\inc$, but the result of this $\inc$ is then ignored, like in $(\lambda x.\zero)\,\inc$.
	Second, a term which takes a productive argument and uses it at least twice is also productive 
	(for example, the productive argument may be a function that creates an $a$-labeled node; 
	when a lambda-term uses such an argument twice, the lambda-term is itself responsible for increasing the number of constants $a$ in the resulting tree).
	
	A \emph{type judgment} is of the form $\Gamma\vdash M:(f,\tau)$, where we require that the type $\tau$ and the term $M$ are of the same sort.
	The \emph{type environment} $\Gamma$ is a set of bindings of variables of the form $x^\alpha:(f,\tau)$, where $\tau\in\Tt^\alpha$.
	In $\Gamma$ we may have multiple bindings for the same variable.
	By $\dom(\Gamma)$ we denote the set of variables $x$ that are bound by $\Gamma$, and
	by $\Gamma\restr_\pr$ we denote the set of those binding from $\Gamma$ that use flag $\pr$.
	%We introduce the order on flags defined by $\np<\pr$, which allows us to use ``$\max$'' and ``$\min$'' operations for flags.

	We now gradually present rules of the type system.
	We begin with rules for node constructors:	
	\begin{mathpar}
	\inferrule{}{\vdash \inc:(\pr,(f,\r)\arr\r)}\and
	\inferrule{}{\vdash \nd:(\np,(f_1,\r)\arr(f_2,\r)\arr\r)}\and
	\inferrule{}{\vdash \zero:(\np,\r)}
	\end{mathpar}
	Since we aim at counting constants $a$, we say here that $a$ is productive, while $\nd$ and $\zero$ are nonproductive.
	Notice that productivity of a node constructor does not depend on productivity of the argument; flags of the arguments ($f, f_1, f_2$) can be arbitrary.

	Then we have a rule for a variable:
	\begin{mathpar}
	\inferrule{}{x:(f,\tau)\vdash x:(\np,\tau)}
	\end{mathpar}
	The type of the variable is taken from the environment. 
	The flag is always $\np$, though; by just using a variable we are not productive at all 
	(and in the productivity flag we want to cover productivity of the lambda-term itself, not of lambda-terms that may be potentially substituted for free variables).
	
	The rule that talks about lambda-binders is very natural; it just moves type pairs from the argument to the environment:
	\begin{mathpar}
	\inferrule*[right=($\lambda$)]{\Gamma\cup\{x:(f_i,\tau_i)\mid i\in I\}\vdash K:(f,\tau)\\x\not\in \dom(\Gamma)}
		{\Gamma\vdash\lambda x.K:(f,\bigwedge\nolimits_{i\in I}(f_i,\tau_i)\arr\tau)}
	\end{mathpar}
	
	Finally, we have the most complicated rule, for application:
	\begin{mathpar}
	\inferrule*[right=$(@)$]{\Gamma\vdash K:(f',\bigwedge\nolimits_{i\in I}(f_i^\bullet,\tau_i)\arr\tau)\\
		\Gamma_i\vdash L:(f_i^\circ,\tau_i)\mbox{ for each }i\in I}
		{\Gamma\cup\bigcup\nolimits_{i\in I}\Gamma_i\vdash K\,L:(f,\tau)}
	\end{mathpar}
	where we assume that 
	\begin{itemize}
	\item	every pair $(f_i^\bullet,\tau_i)$ is different (where $i\in I$),
	\item	for each $i\in I$, $f_i^\bullet=\pr$ if and only if $f_i^\circ=\pr$ or $\Gamma_i\restr_\pr\neq\emptyset$, and
	\item	$f=\pr$ if and only if $f'=\pr$, or $f_i^\circ=\pr$ for some $i\in I$, or $|\Gamma\restr_\pr|+\sum_{i\in I}|\Gamma_i\restr_\pr|>|(\Gamma\cup\bigcup_{i\in I}\Gamma_i)\restr_\pr|$.
	\end{itemize}

	Let us explain the above conditions.
	The first condition is technical: we need to provide exactly one derivation for every needed type pair.
	The second condition says that when $K$ requires a ``productive'' argument, either we can apply an argument $L$ that is itself productive, 
	or we can apply a nonproductive $L$ that uses a productive variable; in the latter case, after substituting something for the variable, $L$ will become productive.
	The third condition says that $K\,L$ is productive if $K$ is productive, or if $L$ is productive, or if some productive free variable is duplicated 
	(i.e., used in at least two subderivations simultaneously).
	
	Notice that weakening of type environments is disallowed: $\Gamma\vdash M:(f,\tau)$ does not necessarily imply $\Gamma,x:(g,\sigma)\vdash M:(f,\tau)$;
	in other words, every binding $x:(g,\sigma)$ in the type environment (and thus every pair $(g,\sigma)$ assigned to an argument) has to be really used somewhere in the type derivation.
	This property of the type system is very expected, if we recall that we want to distinguish lambda-terms that really use their (productive) arguments from those in which the arguments are discarded.
	On the other hand, contraction is allowed: we may say that $\Gamma,x:(g,\sigma),x:(g,\sigma)\vdash M:(f,\tau)$ implies $\Gamma,x:(g,\sigma)\vdash M:(f,\tau)$,
	since a type environment is a set of type bindings.
	As we see in the \AppR rule, such contractions (for productive type binding) cause productivity of lambda-terms.
	
	A \emph{derivation} is defined as usual: it is a tree labeled by type judgments, such that each node together with its children fit to one of the rules of the type system.
	
	We now define a \emph{value} of every node of a derivation, saying how much this node is productive.
	In a node using the rule for the constant $a$, the value is $1$.
	In a node using the \AppR rule with type environments $\Gamma$ and $\Gamma_i$ for $i\in I$, the value is 
	\begin{align*}
		|\Gamma\restr_\pr|+\sum\nolimits_{i\in I}|\Gamma_i\restr_\pr|-|(\Gamma\cup\bigcup\nolimits_{i\in I}\Gamma_i)\restr_\pr|\,.
	\end{align*}
	Spelling this out, the value in such a node equals the number of productive type bindings together in all the type environments $\Gamma$, $(\Gamma_i)_{i\in I}$,
	minus the number of such type bindings in their union.
	In other words, it says how many times we have to duplicate some productive type bindings before splitting them between type environments of subderivations.
	In all other nodes the value is $0$.

	For a derivation $D$, the \emph{value} of $D$, denoted $\low(D)$, is the sum of values of all nodes in $D$.
	We can easily see that the value of a derivation $D$ is positive if and only if $D$ is productive (i.e., the flag in the derived type judgment is $\pr$).
	The main theorem says that $\low(D)$ can be used to estimate the the number of constants $a$ in normal forms of lambda-terms.
		
	\begin{theorem}\label{thm:det}
		The following holds for the type system introduced above:
		\begin{compactenum}[(D1)]
		\item\label{it:d1} 
			for every $m\in\Nat$ there is a function $\eta_m\colon\Nat\to\Nat$ such that 
			if $M$ is a homogeneous and closed lambda-term of sort $\otyp$ and complexity at most $m$, and $D$ is a derivation for $\vdash M:(f,\r)$,
			then the number of constants $a$ in the normal form of $M$ is 
			\begin{compactenum}
			\item[(D1A)]\label{it:d1a} 
				at least $\low(D)$, and 
			\item[(D1B)]\label{it:d1b} 
				at most $\eta_m(\low(D))$;
			\end{compactenum}
		\item\label{it:d2} 
			for every closed lambda-term $M$ of sort $\otyp$ one can derive $\vdash M:(f,\r)$ (for some $f\in\{\pr,\np\}$).\footnote{%
				Actually, one can even prove that there is a unique derivation concerning $M$ (assuming that $M$ is closed and of sort $\otyp$).}		
		\end{compactenum}
	\end{theorem}

	\begin{example}\label{ex:1}
		%We show  
		Observe how the type system behaves for the lambda-term $M=(\lambda y.N\,(N\,(N\,y))\,(a\,e))\,a$, where $N=\lambda y.\lambda x.y\,(y\,x)$.
%		Consider the lambda-term $M=(\lambda y.N\,(N\,(N\,y))\,(a\,e))\,a$ with $N=\lambda y.\lambda x.y\,(y\,x)$,
%		and 
		We start with a derivation concerning $N$, where we write $\tau_y^\pr$ for $(\pr,\r)\arr\r$: % and $\tau_y^\np$ for $(\np,\r)\arr\r$:
		%we write $b_x$ for the binding $x:(\pr,\r)$ and $b_y$ for the binding $y:(\pr,(\pr,\r)\arr\r)$:
		\begin{mathpar}
		\inferrule*[Right=$(\lambda)$,leftskip=-7.7em,rightskip=-13em]{
			\inferrule*[Right=$(\lambda)$,leftskip=1.1em,rightskip=1.1em]{
				\inferrule*[Right=$(@)$,leftskip=6.6em,rightskip=6.6em]{
					y:(\pr,\tau_y^\pr)\vdash y:(\np,\tau_y^\pr)
				\and
					\inferrule*[right=$(@)$,leftskip=1em,rightskip=5.3em]{
						y:(\pr,\tau_y^\pr)\vdash y:(\np,\tau_y^\pr)
					\and
						x:(\pr,\r)\vdash x:(\np,\r)
					}{
						y:(\pr,\tau_y^\pr),\,x:(\pr,\r)\vdash y\,x:(\np,\r)
					}
				}{
					y:(\pr,\tau_y^\pr),\,x:(\pr,\r)\vdash y\,(y\,x):(\pr,\r)
				}
			}{
				y:(\pr,\tau_y^\pr)\vdash\lambda x.y\,(y\,x):(\pr,\tau_y^\pr)
			}
		}{
			\vdash N:(\pr,(\pr,\tau_y^\pr)\arr\tau_y^\pr)
		}
		\end{mathpar}
		Notice that the type $\tau_y^\pr$ requires a productive argument, but (in both the \AppR rules above) we apply an argument that is not productive itself.
		This is possible, because the type judgments for the arguments have productive type bindings in the type environments 
		(and hence for the purposes of the \AppR rule they are assumed to be productive).
		The lower use of the \AppR rule has value $1$ (and in effect the productivity flag is set to $\pr$), 
		because the productive type binding $y:(\pr,\tau_y^\pr)$ is taken to both children.

		Below, we have another derivation concerning $N$, where we write $\tau_y^\np$ for $(\np,\r)\arr\r$:
		\begin{mathpar}
		\inferrule*[Right=$(\lambda)$,leftskip=-5.2em,rightskip=-10.5em]{
			\inferrule*[Right=$(\lambda)$,leftskip=1.1em,rightskip=1.1em]{
				\inferrule*[Right=$(@)$,leftskip=4.1em,rightskip=4.1em]{
					y:(\pr,\tau_y^\pr)\vdash y:(\np,\tau_y^\pr)
				\and
					\inferrule*[right=$(@)$,leftskip=1em,rightskip=5.3em]{
						y:(\pr,\tau_y^\np)\vdash y:(\np,\tau_y^\np)
					\and
						x:(\np,\r)\vdash x:(\np,\r)
					}{
						y:(\pr,\tau_y^\np),\,x:(\np,\r)\vdash y\,x:(\np,\r)
					}
				}{
					y:(\pr,\tau_y^\pr),\,y:(\pr,\tau_y^\np),\,x:(\np,\r)\vdash y\,(y\,x):(\np,\r)
				}
			}{
				y:(\pr,\tau_y^\pr),\,y:(\pr,\tau_y^\np)\vdash\lambda x.y\,(y\,x):(\np,\tau_y^\np)
			}
		}{
			\vdash N:(\np,(\pr,\tau_y^\pr)\wedge(\pr,\tau_y^\np)\arr\tau_y^\np)
		}
		\end{mathpar}
		This time the value of all nodes is $0$, because every type binding is used in exactly one place.
		Likewise, it is possible to derive five other type pairs for the lambda-term $N$:
		\begin{align*}
		&(\np,(\pr,\tau_y^\pr)\wedge(\np,\tau_y^\pr)\arr\tau_y^\pr)\,,&
		&(\np,(\np,\tau_y^\pr)\arr\tau_y^\pr)\,,\\
		&(\np,(\pr,\tau_y^\np)\wedge(\np,\tau_y^\np)\arr\tau_y^\np)\,,&
		&(\np,(\np,\tau_y^\np)\arr\tau_y^\np)\,.\\
		&(\np,(\np,\tau_y^\pr)\wedge(\pr,\tau_y^\np)\arr\tau_y^\np)\,,&
		\end{align*}
		While deriving a type for $M$, we only need one type pair for $N$: the type pair $(\pr,(\pr,\tau_y^\pr)\arr\tau_y^\pr)$ derived at the beginning.
		But we remark that if the lambda-term was $M'=(\lambda y.N\,(N\,(N\,y))\,e)\,a$ 
		(we have replaced here $a\,e$ by $e$, and thus the first call to $N$ receives a nonproductive argument as $x$), 
		it would be necessary to use both the above derivations for $N$.

		Denoting the type $(\pr,\tau_y^\pr)\arr\tau_y^\pr$ as $\tau_N$, we continue the derivation for $M$:		
		\begin{mathpar}
		\inferrule*[Right=$(@)$,rightskip=-9.5em]{
			\vdash N:(\pr,\tau_N)
		\and
			\inferrule*[Right=$(@)$,leftskip=1em,rightskip=4.2em]{
				\vdash N:(\pr,\tau_N)
			\and
				\inferrule*[right=$(@)$,leftskip=1em,rightskip=5.3em]{
					\vdash N:(\pr,\tau_N)
				\and
					y:(\pr,\tau_y^\pr)\vdash y:(\np,\tau_y^\pr)
				}{
					y:(\pr,\tau_y^\pr)\vdash N\,y:(\pr,\tau_y^\pr)
				}
			}{
				y:(\pr,\tau_y^\pr)\vdash N\,(N\,y):(\pr,\tau_y^\pr)
			}
		}{
				y:(\pr,\tau_y^\pr)\vdash N\,(N\,(N\,y)):(\pr,\tau_y^\pr)
		}
		\and
		\inferrule*[right=$(@)$,leftskip=-3.5em]{
			\inferrule*[right=$(\lambda)$,rightskip=1em]{
				\inferrule*[Right=$(@)$,leftskip=5em,rightskip=5em]{
					y:(\pr,\tau_y^\pr)\vdash N\,(N\,(N\,y)):(\pr,\tau_y^\pr)
				\and
					\inferrule*[right=$(@)$,leftskip=1em,rightskip=5.5em]{
						\vdash a:(\pr,\tau_y^\np)
					\and
						\vdash e:(\np,\r)
					}{
						\vdash a\,e:(\pr,\r)
					}
				}{
					y:(\pr,\tau_y^\pr)\vdash N\,(N\,(N\,y))\,(a\,e):(\pr,\r)
				}
			}{
				\vdash\lambda y.N\,(N\,(N\,y))\,(a\,e):(\pr,(\pr,\tau^\pr_y)\arr\r)
			}
		\and
			\vdash a:(\pr,\tau_y^\pr)
		}{
			\vdash M:(\pr,\r)
		}
		\end{mathpar}
		
		The total value of this derivation is $5$ ($2$ in the two nodes concerning $a$, and $3$ in the three subderivations concerning $N$), 
		while the normal form of $M$ contains $9$ appearances of the constant $a$.
		Notice that while adding any further $N$ to the sequence $N\,(N\,(N\,y))$, we increase the value by $1$, while we almost double the number of $a$'s in the normal form.
%		
%		We remark that it is possible 
%		\begin{mathpar}
%		\inferrule*[Right=$(@)$,rightskip=-6.5em]{
%			b_y\vdash y:(\np,(\pr,\r)\arr\r)\\
%			\inferrule*[right=$(@)$,leftskip=1em,rightskip=6.5em]{
%				b_y\vdash y:(\np,(\pr,\r)\arr\r)\\
%				b_z\vdash z:(\np,\r)
%			}{b_z,\ b_y\vdash y\,z:(\np,\r)}
%		}{b_z,\ b_y\vdash y\,(y\,z):(\pr,\r)}
%		\and
%		\inferrule*[Right=$(\lambda)$,leftskip=-6em,rightskip=-13em]{
%			\inferrule*[Right=$(@)$,leftskip=6em,rightskip=6em]{
%				\vdash\inc:(\pr,(\pr,\r)\arr\r)\\
%				\inferrule*[right=$(@)$,leftskip=1em,rightskip=7em]{
%					\vdash\inc:(\pr,(\pr,\r)\arr\r)\\
%					x:(\pr,\r)\vdash x:(\np,\r)
%				}{x:(\pr,\r)\vdash\inc\,x:(\pr,\r)}
%			}{x:(\pr,\r)\vdash\inc\,(\inc\,x):(\pr,\r)}
%		}{\vdash\lambda x.\inc\,(\inc\,x):(\pr,(\pr,\r)\arr\r)}
%		\end{mathpar}
%		The duplication factor of the root node of the first derivation is $1$, because the binding for $y$ is used in both subderivations;
%		the other nodes have duplication factors $0$. 
%	
%		It is possible to derive six other type judgments containing the term $y\,(y\,z)$:
%		\begin{align*}
%		&y:(\np,(\pr,\r)\arr\r),\,z:(\pr,\r)\vdash y\,(y\,z):(\np,\r),\\
%		&y:(\pr,(\pr,\r)\arr\r),\,y:(\np,(\pr,\r)\arr\r),\,z:(\pr,\r)\vdash y\,(y\,z):(\np,\r),\\
%		&y:(\pr,(\pr,\r)\arr\r),\,y:(\pr,(\np,\r)\arr\r),\,z:(\np,\r)\vdash y\,(y\,z):(\np,\r),\\
%		&y:(\pr,(\np,\r)\arr\r),\,y:(\np,(\np,\r)\arr\r),\,z:(\np,\r)\vdash y\,(y\,z):(\np,\r),\\
%		&y:(\np,(\pr,\r)\arr\r),\,y:(\pr,(\np,\r)\arr\r),\,z:(\np,\r)\vdash y\,(y\,z):(\np,\r),\\
%		&y:(\np,(\np,\r)\arr\r),\,z:(\np,\r)\vdash y\,(y\,z):(\np,\r).
%		\end{align*}
	\end{example}
		
	\paragraph*{Proofs.}
	Let us now sketch the proof of Theorem~\ref{thm:det}.
	While proving Condition~\hyperref[it:d1]{(D1)}, it is convenient to consider the RMF strategy of reductions (defined on Page~\pageref{pag:rmf}).
	We have the following subject-reduction lemma for reductions of this kind.
	
	\begin{lemma}\label{lem:subj-red-det}
		If $D_0$ is a derivation for $\vdash M_0:(f,\otyp)$, where $M_0$ is homogeneous, closed, and of complexity $m$ (and of sort $\otyp$), 
		and $M_0\redb M_1\redb\dots\redb M_n$ is a sequence of RMF$(m)$ beta-reductions,
		then there exists a derivation $D_n$ for $\vdash M_n:(f,\otyp)$ such that $\low(D_0)\leq\low(D_n)$ and $\low(D_n)\leq 2^{\low(D_0)}$.
	\end{lemma}

	Because the maximal complexity $m$ of the lambda-term $M$ considered in Theorem~\ref{thm:det} is fixed, 
	using Lemma~\ref{lem:subj-red-det} $m$ times (for complexities $m,m-1,\dots,1$) we obtain a derivation $D_T$ for the normal from $T$ of $M$ 
	such that $\low(D)\leq\low(D_T)$ and $\low(D_T)$ is bounded by a function of $\low(D)$,
	that is, $\low(D)$ estimates $\low(D_T)$.
	It remains to notice that $\low(D_T)$ is exactly the number of $a$-labeled nodes in the tree $T$.
	
	\begin{proof}[Proof sketch (Lemma~\ref{lem:subj-red-det})]
		We proceed by induction: for every $i\in\{1,\dots,n\}$ out of the derivation $D_{i-1}$ for $\vdash\nobreak M_{i-1}:(f,\otyp)$ we construct a derivation $D_i$ for $\vdash M_i:(f,\otyp)$.
		To this end, we consider every subderivation $D$ of $D_{i-1}$ starting with a type judgment $\Gamma\vdash(\lambda x.K)\,L:(g,\tau)$ concerning the redex involved in the reduction $M_{i-1}\redb M_i$;
		we need to replace it by a derivation $D'$ for $\Gamma\vdash K[L/x]:(g,\tau)$.
		We obtain $D'$ by a surgery on $D$:
		we take the subderivation of $D$ concerning $K$, we replace every leaf deriving a type $\sigma$ for $x$ by the subderivation of $D$ deriving this type $\sigma$ for $L$,
		and we update type environments and productivity flags appropriately.

		Notice that every subderivation concerning $L$ is moved to at least one leaf concerning $x$ (nothing can disappear).
		The only reason why the value of the derivation can decrease is that potentially a productive type binding $x:(\pr,\sigma)$ was duplicated (say, $k$ times) in the derivation concerning $K$.
		In $D'$ this binding is no longer present (in $K[L/x]$ there is no $x$) so the value gets decreased by $k$,
		but in this situation the subderivation deriving $\sigma$ for $L$ becomes inserted in $k+1$ leaves.
		This subderivation is either productive itself, or uses a productive type binding in the environment;
		in both cases by creating $k$ additional copies of this subderivation we increase the value at least by $k$, compensating the loss caused by elimination of $x$.
		This implies that $\low(D)\leq\low(D')$, hence $\low(D_{i-1})\leq\low(D_i)$ (and, in effect, $\low(D_0)\leq\low(D_n)$).
		
		Conversely, the only reason why the value can grow is that some derivation concerning $L$ (that is either productive itself or uses some productive type bindings for its free variables)
		becomes inserted in $k+1$ leaves, for some $k\geq 1$.
		In the worst case, this may cause that the value (of the whole derivation for $M$) gets multiplied by $k+1$.
		But, simultaneously, in the subderivation concerning $K$, the productive type bindings for $x$ are removed,
		which decreases the value by $k$ in some nodes of this subderivation.
		The point is now that these nodes were never copied in the reduction sequence from $D_0$ to the considered $D_{i-1}$;
		this is because all the reductions are RMF$(m)$ reductions.
		Indeed, looking from the other side, all variables appearing in (the copied subderivation for) $L$ are of order at most $m-2$---%
		as observed on Page~\pageref{pag:rmf}---%
		but all variables involved in future order-$m$ reductions (i.e., all variables that we remove from type environments) are of order $m-1$---%
		because of homogeneity of the lambda-term.
		Thus, whenever we multiply the value of the current derivation by at most $k+1$,
		we subtract $k$ from the value of the original derivation $D_0$.
		The worst case is when $\low(D_0)$ times we decrease the value by $1$, and $\low(D_0)$ times we multiply it by $2$.
		It follows that $\low(D_n)\leq\low(D_0)\cdot2^{\low(D_0)}$; a slightly more careful analysis shows that actually $\low(D_n)\leq 2^{\low(D_0)}$.
	\end{proof}
	
	In the proof of Condition~\hyperref[it:d2]{(D2)}, saying that we can derive a type for every closed lambda-term $M$ of sort $\otyp$, we proceed backwards:
	it is easy to derive a type for a tree (i.e., for the normal form of $M$), 
	and thus it is enough to have a subject expansion lemma 
	saying that out of a derivation for a lambda-term after a beta-reduction we can construct a derivation for the lambda-term before the beta-reduction.
	This time we follow the OI reduction strategy.
	Because outermost redexes are closed, it is thus enough to have the following lemma.
	
	\begin{lemma}\label{lem:subj-exp}
		If we can derive $\vdash K[L/x]:(g,\tau)$, then we can also derive $\vdash (\lambda x.K)\,L:(g,\tau)$.
	\end{lemma}

	\begin{proof}[Proof sketch]
		In the derivation $D$ for $K[L/x]$ we replace every subderivation concerning $L$ by a leaf rule for the variable $x$,
		and we correct type environments and productivity flags in the rest of the derivation.
		This way we obtain a derivation for $K$ with type environment requesting some types for $x$.
		Simultaneously, each of these types was derived for $L$ in some subderivation of $D$ 
		(there may be multiple such subderivations, because $L$ may appear in many places in $K[L/x]$, but we choose only one subderivation for every type).
		It is not difficult to combine these derivations into a derivation concerning $(\lambda x.K)\,L$.
	\end{proof}

	We remark that by applying the above surgery to a derivation for $\Gamma\vdash K[L/x]:(g,\tau)$ 
	(i.e., for an arbitrary redex, having some free variables) we only obtain a derivation for $\Gamma'\vdash (\lambda x.K)\,L:(g,\tau)$ with some $\Gamma'\subseteq\Gamma$,
	but not necessarily with $\Gamma'=\Gamma$.
	The reason is that we remove some subderivations concerning $L$ (we leave only one for every type), 
	and possibly some type bindings from $\Gamma$ were used only in the removed subderivations.

	\paragraph*{Bibliographic Note.}
	As already mentioned in the introduction, the idea of the type system presented above originates from Parys~\cite{ho-new}.
	In that paper, a similar type system was introduced for configurations of collapsible pushdown systems.
	It was then used to prove that a restricted variant of these systems (systems without the so-called collapse operation) are less powerful than general collapsible pushdown systems.
	The type system was then transferred to the setting of lambda-terms in Parys~\cite{numbers-journal}.
	Their type system is slightly more complicated than ours, and allows to obtain a stronger version of Condition~\hyperref[it:d1]{(D1B)},
	where the function $\eta_m$ does not depend on the complexity $m$ of considered lambda-terms.

\section{Nondeterministic Quantities}\label{sec:nondeterministic}

	Suppose now that we want to estimate another quantity: the maximal number of appearances of the constant $\inc$ on a single branch in the beta-normal form $T$ of a lambda-term $M$.
	It seems that in order to describe this quantity, it is enough to take the type system from Section~\ref{sec:deterministic}, and replace the rule for the constant $b$ by two rules:
	\begin{mathpar}
	\inferrule{}{\vdash \nd:(\np,(f,\r)\arr\top\arr\r)}\and
	\inferrule{}{\vdash \nd:(\np,\top\arr(f,\r)\arr\r)}
	\end{mathpar}
	In these rules we ignore one of the arguments, and we descend only to the other one.
	This way, every type derivation $D$ for a tree $T$ follows one branch in $T$, and in effect $\low(D)$ equals to the number of constants $a$ on that branch.
	By arguments like in the previous section we obtain the following, rather useless, properties of the modified type system:
	\begin{compactenum}[(N1)]
	\item\label{it:n1} 
		for every $m\in\Nat$ there is a function $\eta_m\colon\Nat\to\Nat$ such that 
		if $M$ is a homogeneous and closed lambda-term of sort $\otyp$ and complexity at most $m$, and $D$ is a derivation for $\vdash M:(f,\r)$,
		then the number of constants $a$ on some branch of the normal form of $M$ is
		\begin{compactenum}
		\item[(N1A)]\label{it:n1a} 
			at least $\low(D)$, and 
		\item[(N1B)]\label{it:n1b} 
			at most $\eta_m(\low(D))$;
		\end{compactenum}
	\item\label{it:n2} 
		for every closed lambda-term $M$ of sort $\otyp$ one can derive $\vdash M:(f,\r)$ (for some $f\in\{\pr,\np\}$).
	\end{compactenum}

	These properties are not satisfactory for us, because they only say that there exists a branch with the number of constants $a$ estimated by $\low(D)$, for some derivation $D$.
	We, however, are interested in the branch on which the number of constants $a$ is maximal.
	In other words: if in the beta-normal form $T$ of $M$ there are two branches, one with just a few constants $a$, and the other with a lot of them,
	we expect to have two derivations $D$ and $D'$, where $\low(D)$ is small (corresponds to the first branch), and $\low(D')$ is large
	(corresponds to the second branch).
	But Condition~\hyperref[it:n2]{(N2)} gives us only one derivation, and we do not know which one.
	Thus, we rather need to have the following property:
	\begin{compactenum}
	\item[(N2$'$)\!]\label{it:n2p} 
		for every $m\in\Nat$ there is a function $\eta_m\colon\Nat\to\Nat$ such that if $M$ is a homogeneous and closed lambda-term of sort $\otyp$ and complexity at most $m$
		and on some branch of the beta-normal form of $M$ there are $n$ appearances of the constant $a$,
		then there is a derivation $D$ for $\vdash M:(f,\r)$ such that $n\leq\eta_m(\low(D))$.
	\end{compactenum}

	In the light of Condition~\hyperref[it:n2p]{(N2$'$)}, Condition~\hyperref[it:n1b]{(N1B)} becomes redundant, and thus we can restate Condition~\hyperref[it:n1a]{(N1A)} as follows:
	\begin{compactenum}
	\item[(N1$'$)\!]\label{it:n1p} 
		if $M$ is a homogeneous and closed lambda-term of sort $\otyp$, and $D$ is a derivation for $\vdash M:(f,\r)$,
		then the number of constants $a$ on some branch of the normal form of $M$ is at least $\low(D)$.
	\end{compactenum}
	
	It is, though, an open problem whether Condition~\hyperref[it:n2p]{(N2$'$)} holds.
	
	\begin{openpr}
		Does the modified type system satisfy Condition~\hyperref[it:n2p]{(N2$'$)}?
	\end{openpr}

	In order to prove Condition~\hyperref[it:n2p]{(N2$'$)}, we should probably proceed backward:
	we should start with a derivation concerning (the branch with the maximal number of constants $a$ in) the normal form of $M$,
	and then, successively, from a derivation for a lambda-term after a beta-reduction obtain a derivation for the lambda-term before the beta-reduction.
	We have a subject expansion lemma (Lemma~\ref{lem:subj-exp}) only for redexes without free variables 
	(and it seems difficult to generalize it to arbitrary redexes, as explained at the end of the previous section);
	we should thus assume that we always reduce the outermost redex.
	In effect, in the considered sequence of beta-reductions from $M$ to its normal form we have to mix reductions concerning redexes of different orders.
	For such a sequence of reductions it is not clear how to estimate the value of the derivation for the beta-normal form $T$ by the value of the derivation for $M$.
	
	We remark that a modified type system, in which one allow weakening of type environments, satisfies a subject expansion lemma (like Lemma~\ref{lem:subj-exp}).
	But with unrestricted weakening of type environments Condition~\hyperref[it:n1p]{(N1$'$)} no longer holds.
	Indeed,	if weakening was allowed, we could use a derivation $D$ (with an arbitrary large value) for a lambda-term $M$ as a part of a derivation for a lambda-term like $(\lambda x.\zero)\,M$,
	whose normal form contains no $\inc$.
	The reason why weakening is forbidden is exactly this: we want to have subderivations only for subterms that really participate to the normal form.

	The life is thus not so simple: because we want both Conditions~\hyperref[it:n1p]{(N1$'$)} and~\hyperref[it:n2p]{(N2$'$)}, we have to introduce a more complicated type system.
	In this type system, instead of one kind of values of nodes, we have \emph{values of order $k$} (or \emph{$k$-values}) for every $k\in\{1,\dots,m+1\}$ 
	(where $m$ is the complexity of the considered lambda-term).
	We also mark some nodes as belonging to a \emph{zone of order $k$} (or \emph{$k$-zone}) for every order $k\in\{0,\dots,m\}$.

	Before defining the type system, let us first give some idea how Condition~\hyperref[it:n2p]{(N2$'$)} can be shown.
	Then, we give details of a type system motivated by this idea.

	Consider thus a lambda-term $M_m$ that is of complexity $m$, and reduces to a tree $M_0$.
	Following the RMF reduction strategy, we can find lambda-terms $M_{m-1},M_{m-2},\dots,M_1$ such that every $M_i$ is of complexity $i$ and all reductions between $M_i$ and $M_{i-1}$ are of order $i$.
	Our aim is to estimate the number of constants $a$ located on some branch in $M_0$.
	We thus mark all nodes of this branch as the $0$-zone, and we say that the order-$1$ value is $1$ in all nodes of the $0$-zone that are labeled by $a$.
	Next, we proceed back to $M_1$.
	Every node constructor in the $0$-zone in $M_0$ originates from some particular node constructor appearing already in $M_1$.
	We thus mark these node constructors in $M_1$ as belonging to the $0$-zone (notice that in $M_1$ they no longer form a branch);
	and again those of them that are $a$-labeled get $1$-value $1$.
	The crucial observation is that no two node constructors from the $0$-zone in $M_0$ can originate from a single node constructor of $M_1$.
	Indeed, all the beta-reductions between $M_1$ and $M_0$ are RMF$(1)$.
	In such a beta-reduction we take a whole subtree (i.e., a lambda-term of sort $\otyp$) of $M_1$, and we replace it somewhere, possibly replicating it.
	But since the considered nodes of $M_0$ lie on a single branch, they may belong to at most one copy of the replicated subtree.
	In effect, the total $1$-value in $M_1$ is the same as in $M_0$.
	
	We cannot directly repeat the same reasoning to move $1$-values from $M_1$ back to $M_2$, since now there is a problem: 
	a single node constructor in $M_2$ may result in multiple (uncontrollably many) node constructors with a $1$-value in $M_1$.
	We rescue ourselves in the following way.
	We choose some branch of $M_1$ (included in the $0$-zone) as the $1$-zone.
	Then, for every node of $M_1$ with positive $1$-value, we look for the closest ancestor of this node that lies in the $1$-zone,
	and in this ancestor we set the $2$-value to $1$.
	Notice that for multiple nodes with positive $1$-value, their closest ancestor lying in the $1$-zone may be the same 
	(and then we set its $2$-value to $1$, not to the number of these nodes).
	Thus, in general, the total $2$-value may be smaller than the total $1$-value.
	We can, however, ensure that it is smaller only logarithmically; 
	to do so, we choose a branch forming the $1$-zone in a clever way: staring from the root, we always proceed to the subtree with the largest total $1$-value.
	In effect, the total $2$-value of $M_1$ estimates the total $1$-value of $M_1$.

	Once all nodes of $M_1$ with positive $2$-value lie on a single branch (which is chosen as the $1$-zone), we can transfer them back to $M_2$ without changing their number:
	because reductions between $M_2$ and $M_1$ are RMF$(2)$, every node of the $1$-zone in $M_1$ originates from a different node of $M_2$.
	Then in $M_2$ we again choose a branch as the $2$-zone, we assign $3$-value to some its nodes, and so on.
	At the end we obtain some labeling of $M_m$ by zones and values of particular orders.
	The goal of the type system presented below is, roughly speaking, to ensure that a labeling of $M_m$ actually is obtainable in the process as above.
	In fact, we do not label nodes of $M_m$ itself, but rather nodes of a type derivation for $M_m$.

	We now come to a formal definition of the type system.

	\paragraph*{Type Judgments.}
	
	For every sort $\alpha$ we define the set $\Tt^\alpha$ of \emph{types} of sort $\alpha$,
	and the set $\Ttrip^\alpha_m$ of \emph{type triples} of sort $\alpha$.
	This is done as follows, where $\Pp$ denotes the powerset:
	\begin{align*}
		&\Tt^{\alpha\arr\beta}=\Pp(\Ttrip_{\ord(\alpha)}^\alpha)\times\Tt^\beta\,,\qquad
		\Tt^\otyp=\{\otyp\}\,,\\
		&\Ttrip_m^\alpha=
			%\bigcup{}_{B\in\{0,\dots,m\}}\times\{\emptyset,\{B+1\}\}\times\Tt^\alpha\,.
			\{(Z,F,\tau)\in\{0,\dots,m\}^2\times\Tt^\alpha\mid F\leq Z+1\}\,.
	\end{align*}
	Notice that the sets $\Tt^\alpha$ and $\Ttrip_m^\alpha$ are finite.
	A type $(T,\tau)\in\Tt^{\alpha\arr\beta}$ is denoted as $T\arr\tau$.
	A type triple $\hat\tau=(Z,F,\tau)\in\Ttrip_m^\alpha$ consists of a zone order $Z$, a productivity order $F$, and a type $\tau$.
	In order to distinguish types from type triples, the latter are denoted by letters with a hat, like $\hat\tau$.
	
	A \emph{type judgment} is of the form $\Gamma\vdash_m M:\hat\tau$, where $\Gamma$, called a \emph{type environment}, 
	is a set of bindings of the form $x^\alpha:\hat\sigma$ with $\hat\sigma\in\Ttrip^\alpha_{\ord(\alpha)}$,
	and $M$ is a lambda-term, and $\hat\tau$ is a type triple of the same sort as $M$ (i.e., $\hat\tau\in\Ttrip^\beta_m$ when $M$ is of sort $\beta$).
	We assume that $M$ is homogeneous.
	
	As previously, the intuitive meaning of a type $\bigwedge T\arr\tau$ is that a lambda-term having this type can return a lambda-term having type $\tau$, while taking an argument for which we can derive all type triples from $T$.
	Moreover, in $\Tt^\otyp$ there is just one type $\otyp$, which can be assigned to every lambda-term of sort $\otyp$.
	Suppose that a node of a type derivation for a closed and homogeneous lambda-term $M_m$ of sort $\otyp$ is labeled by a type judgment $\Gamma\vdash_m M:\hat\tau$ with $\hat\tau=(Z,F,\tau)$.
	Then
	\begin{itemize}
	\item	$\tau$ is the type derived for $M$;
	\item	$\Gamma$ contains type triples that could be used for free variables of $M$ in the derivation;
	\item	$m$ is an upper bound for the complexity of $M$ (this bound is not strict: in the proofs, it is useful to temporarily allow also lambda-terms $M$ of complexity $m+1$), 
		and simultaneously for orders of considered zones and values;
	\item	$Z\in\{0,\dots,m\}$ is the largest number such that for every $k\in\{0,\dots,Z\}$, the considered node of the derivation belongs to the $k$-zone;
	\item	$F\in\{0,\dots,m\}$ is the largest number such that for every $k\in\{1,\dots,F\}$, 
		in the imaginary lambda-term $M_k$ obtained from $M_m$ by reducing all redexes of order greater than $k$,
		the order-$k$ value will be positive in the subderivation starting in the considered node.
	\end{itemize}
	
	Notice that we always have that $Z\geq 0$, which means that every node of every derivation belongs at least to the $0$-zone.
	We choose zones in a derivation in such a way that for every node the set of orders $k$ of zones to which the node belongs is always of the form $\{0,\dots,Z\}$.
	For this reason in a type triple it is enough to have a number $Z$ (representing the set $\{0,\dots,Z\}$), instead of an arbitrary set of orders of zones.
	Moreover, if a node of a derivation belongs to a $k$-zone, then its parent as well; in effect, the zone order in the type triple labeling a parent cannot be smaller than in its child.
	Likewise, the set of orders $k$ for which the $k$-value is positive (after appropriate reductions) is always of the form $\{1,\dots,F\}$, so it is enough to remember its maximum.
	Moreover, if $k$-value is positive is some subderivation, then it is also positive in a larger subderivation, 
	hence also the productivity order in the type triple labeling a parent cannot be smaller than in a child.

	\paragraph{Type System.}

	We now give the first four rules, concerning node constructors:
	\begin{align*}
	&\inferrule{}{\vdash_m \nd:(Z,0,(0,0,\r)\arr\top\arr\r)}&
	&\inferrule{}{\vdash_m \inc:(Z,\min(Z+1,m),(0,0,\r)\arr\r)}\\
	&\inferrule{}{\vdash_m \nd:(Z,0,\top\arr(0,0,\r)\arr\r)}&
	&\inferrule{}{\vdash_m \zero:(Z,0,\r)}
	\end{align*}
	
	We say that the $k$-value in a node using the rule for the constant $a$ is $1$ for all $k\in\{1,\dots,Z+1\}$;
	for $k>Z+1$, and for the other constants the $k$-value is $0$.

	In the above rules we can choose $Z$ arbitrarily (from the set $\{0,\dots,m\}$), 
	which amounts to deciding to which zones the node constructor should belong: it belongs to the $k$-zone for $k\in\{0,\dots,Z\}$.
	For the constant $b$ we descend only to one argument (because we want to count constants $a$ only on a single branch of the normal form).
	For the constant $a$ we have set the $k$-value to $1$ for all $k\in\{1,\dots,Z+1\}$, hence we set the productivity order to $Z+1$.
	There is an exception for $Z=m$: by definition of $\Ttrip_m^\alpha$, the productivity order can be at most $m$, 
	so although the $(m+1)$-value is $1$ as well, this information is not covered by the productivity order.
	Notice that the type $(0,0,\r)$ assigned to arguments of node constructors is the only element of $\Ttrip^\otyp_{\ord(\otyp)}$;
	node constructors do not receive information about zones or values from their arguments.
	
	Next, we have a rule for a variable (in nodes using this rule, the $k$-value is $0$ for all $k$):
	\begin{mathpar}
	\inferrule*[right=(Var)]{(Z'=Z) \lor (Z'\geq\ord(x)=Z)}{x:(Z,F,\tau)\vdash_m x:(Z',F,\tau)}
	\end{mathpar}
	
	In order to understand this rule, suppose that it labels a node of a type derivation for a closed lambda-term $M$ of sort $\otyp$.
	Take some $k\in\{0,\dots,m\}$, and consider the lambda-term $M_k$ obtained from our lambda-term by reducing all redexes of orders greater than $k$.
	According to the proof idea presented above, we create the $k$-zone as a branch of $M_k$ (and then we transfer it back to $M$).
	Moreover, as the productivity order we should take at least $k$ if in $M_k$ the $k$-value is positive in the subtree starting in the considered node.
	If $k\leq\ord(x)$, the variable $x$ will be no longer present in $M_k$, and some lambda-term (described by the type environment) will be substituted for it.
	For this reason, the information about the $k$-zone and about positivity of the $k$-value is taken from the type environment.
	Conversely, if $k>\ord(x)$, the node (leaf) concerning $x$ will be still present in $M_k$, 
	and thus we can start the branch forming the $k$-zone there.
	But this is possible only if the node belongs to the $(k-1)$-zone; 
	in particular for $k=\ord(x)+1$ we need to be in the $\ord(x)$-zone, which is the case if $Z=\ord(x)$.
	Moreover, the total $k$-value in (the subtree starting in) the considered leaf is $0$, and thus the productivity order is taken from the environment (unlike in the previous type system).
	
	The rule for lambda-binders realizes a restricted variant of type weakening: we may ignore arguments that do not contain leaves of zones.
	This is formalized in the notion of balanced and unbalanced type triples, defined by induction on their structure.
	For $k\in\{0,\dots,m\}$, a type triple $(Z,F,\bigwedge T_1\arr\dots\arr\bigwedge T_n\arr\otyp)$ is \emph{$k$-unbalanced} if $Z\geq k$ and all elements of the sets $T_1,\dots,T_n$ are $k$-balanced;
	otherwise, the type triple is \emph{$k$-balanced}.
	A type triple is \emph{unbalanced} if it is $k$-unbalanced for some $k\in\{0,\dots,m\}$; otherwise it is \emph{balanced}.
	Intuitively, a subderivation derives a $k$-unbalanced type triple if the unique leaf of the $k$-zone is contained either in this subderivation, 
	or in an imaginary subderivation that will be substituted for a free variable.
	Indeed, the subderivation contains the leaf of the $k$-zone if it belongs to the $k$-zone, but none of the arguments provides the leaf.
	
%	Namely, we can ignore type bindings $x:(Z',F',\sigma)$ such that $Z'<\ord(x)$.
	We can now give the rule; for nodes using this rule, the $k$-value is $0$ for all $k$.
	\begin{mathpar}
	\inferrule*[right=($\lambda$)]{
			\Gamma\cup\{x:\hat\sigma\mid \hat\sigma\in T'\}\vdash_m K:(Z,F,\tau)
		\\
			\{\hat\sigma\in T\mid \hat\sigma\mbox{ unbalanced}\}\subseteq T'\subseteq T
		\\ 
			x\not\in\dom(\Gamma)
		}
		{\Gamma\vdash_m\lambda x.K:(Z,F,\bigwedge T\arr\tau)}
	\end{mathpar}
	
	As previously, the rule for application is the most complicated one:
	\begin{mathpar}
	\inferrule*[right=$(@)$]{
			\Gamma_0\vdash_m K:(Z_0,F_0,\tau_0)
		\\
			\Gamma_i\vdash_m L:(Z_i,F_i,\tau_i)\mbox{ for each }i\in I
		\\
			\tau_0=\bigwedge\nolimits_{i\in I}(\min(Z_i,\ord(L)),\min(F_i,\ord(L)),\tau_i)\arr\tau
		\and
			Z = \max\nolimits_{i\in\{0\}\cup I}Z_i
		\\ 
			\forall k\in\{0,\dots,m\}.\,|\{i\in\{0\}\cup I\mid (Z_i,F_i,\tau_i)\mbox{ $k$-unbalanced}\}|\leq 1
		}
		{\bigcup\nolimits_{i\in \{0\}\cup I}\Gamma_i\vdash_m K\,L:(Z,F,\tau)}
	\end{mathpar}
	where 
	\begin{itemize}
	\item	we assume that $0\not\in I$;
	\item	if there is $i\in \{0\}\cup I$ such that $\ord(L)\leq Z_i<F_i\leq Z$, then we set $F$ to $\min(Z+1,m)$,
		and we say that the $k$-value in the node using such a rule is $1$ for all $k\in\{F_i+1,\dots,Z+1\}$
		(if there are multiple such $i$, we consider the one for which $F_i$ is the smallest);
	\item	otherwise we set $F$ to $\max_{i\in\{0\}\cup I}F_i$, and the $k$-value to $0$ for all $k$.
	\end{itemize}

	Let us comment on the above conditions.
	First, notice that to the subderivation concerning $K$ we pass the information about $k$-values and $k$-zones from the subderivations concerning $L$ only for $k\leq\ord(L)$
	(i.e., we write $\min(Z_i,\ord(L))$ and $\min(F_i,\ord(L))$ instead of simply $Z_i$ and $F_i$).
	This is because, while thinking about $k$-values and about the $k$-zone, 
	we should imagine the lambda-term $M_k$ obtained from the lambda-term under consideration by reducing all redexes of orders greater than $k$.
	If $k\leq\ord(L)$, the application (for which we write the rule) is no longer present in $M_k$ (it gets reduced in some of the reductions leading to $M_k$),
	so we should pass the information from $L$ to $K$.
	Conversely, if $k>\ord(L)$, the application is still present in $M_k$; this means $K$ and $L$ are independent subterms there, and hence the information from $L$ should not be passed to $K$.
	This is complementary to what we said on the \VarR rule.
	
	Second, we also say for every $k$ that only one child can be $k$-unbalanced.
	Under the intuitive meaning that a conclusion of a subderivation is $k$-unbalanced if the subderivation contains the leaf of the $k$-zone (that remains a leaf in $M_k$),
	this condition ensures that the $k$-zone has at most one leaf, and thus forms a branch in $M_k$.
	
	Third, observe that the $(k+1)$-value in our node is set to $1$ if, in $M_k$, it is the closest ancestor of some node with positive $k$-value that lies in the $k$-zone.
	Indeed, suppose that the current node is still present in $M_k$ (i.e., that $k>\ord(L)$), and that it belongs to the $k$-zone (i.e., that $k\leq Z$).
	Moreover, suppose that in $M_k$ the $k$-value is positive in some node of the subderivation number $i$ (i.e., that $k\leq F_i$), where $i\in\{0\}\cup I$.
	If $k\leq Z_i$, then the closest ancestor being in the $k$-zone is already in the subderivation (because its root belongs to the $k$-zone).
	Conversely, if $k>Z_i$, the closest ancestor being in the $k$-zone is in our node.
	Recall that (by definition of type triples) we always have $F_i\leq Z_i+1$.
	All the inequalities hold when $\ord(L)+1\leq Z_i+1=k=F_i\leq Z$, and this is exactly the situation when we set the $(k+1)$-value of the current node to $1$.
	If the node is also in the $(k+1)$-zone (i.e., if $k+1\leq Z$), then the closest ancestor being in the $(k+1)$-zone is in the node itself.
	It thus makes sense that we also set the $(k+2)$-value of the current node to $1$.
	Repeating this again, we should set to $1$ the values of all orders in $\{k+1,\dots,Z+1\}$.

	Denoting the $k$-value of a derivation $D$ by $\low^k(D)$, we can state the desired properties of our type system.
		
	\begin{theorem}\label{thm:nondet}
		The following holds for the type system introduced above:
		\begin{compactenum}[(N1$''$)]
		\item\label{it:n1pp} 
			if $M$ is a homogeneous and closed lambda-term of sort $\otyp$, and $D$ is a derivation for $\vdash_m M:(m,m,\r)$,
			then the number of constants $a$ on some branch of the normal form of $M$ is at least $\low^{m+1}(D)$;
		\item\label{it:n2pp} 
			for every $m\in\Nat$ there is a function $\eta_m\colon\Nat\to\Nat$ such that if $M$ is a homogeneous and closed lambda-term of sort $\otyp$ and complexity at most $m$,
			and on some branch of the beta-normal form of $M$ there are $n\geq 1$ appearances of the constant $a$,
			then there is a derivation $D$ for $\vdash_m M:(m,m,\r)$ such that $n\leq\eta_m(\low^{m+1}(D))$.
		\end{compactenum}
	\end{theorem}

	\begin{example}\label{ex:2}
		Let us consider the same lambda-term as in Example~\ref{ex:1}, namely $M=(\lambda y.N\,(N\,(N\,y))\,(a\,e))\,a$ with $N=\lambda y.\lambda x.y\,(y\,x)$.
		As $m$ we take its complexity, that is, $2$.
		Notice that after performing all beta-reductions of order $2$, we obtain the lambda-term $M_1=(\lambda x.N_2\,(N_2\,x))\,(a\,e)$ with $N_2=\lambda x.N_1\,(N_1\,x)$ and $N_1=\lambda x.a\,(a\,x)$.
		In this term, the $1$-zone, which has to be a branch, can descend into one of the subterms $N_2$, then into one of the subterms $N_1$, and then it can finish in one of the constants $a$.
		In effect, while typing $M$, we need two derivations for $N$, one where the lambda-term belongs to the $1$-zone, and one where it does not.
		Denote $\tau_y=(0,0,\r)\arr\r$.
		Outside of the $1$-zone, we only pass (from the argument) the information that the $1$-value is positive:
		\begin{mathpar}
		\inferrule*[Right=$(\lambda)$,leftskip=-8.5em,rightskip=-13.8em]{
			\inferrule*[Right=$(\lambda)$,leftskip=1.6em,rightskip=1.6em]{
				\inferrule*[Right=$(@)$,leftskip=6.9em,rightskip=6.9em]{
					y:(0,1,\tau_y)\vdash_2 y:(0,1,\tau_y)
				\and
					\inferrule*[right=$(@)$,leftskip=1em,rightskip=5.3em]{
						y:(0,1,\tau_y)\vdash_2 y:(0,1,\tau_y)
					\and
						x:(0,0,\r)\vdash_2 x:(0,0,\r)
					}{
						y:(0,1,\tau_y),\,x:(0,0,\r)\vdash_2 y\,x:(0,1,\r)
					}
				}{
					y:(0,1,\tau_y),\,x:(0,0,\r)\vdash_2 y\,(y\,x):(0,1,\r)
				}
			}{
				y:(0,1,\tau_y)\vdash_2\lambda x.y\,(y\,x):(0,1,\tau_y)
			}
		}{
			\vdash_2 N:(0,1,(0,1,\tau_y)\arr\tau_y)
		}
		\end{mathpar}
		Notice that in the second (i.e., lower) node using the \AppR rule, the function of type $\tau_y$, that is $(0,0,\r)\arr\r$, accepts an argument with type triple $(0,1,\r)$.
		This is correct, because according to the \AppR rule, the function receives the information only about zones and values of order not greater than the order of the argument, 
		which is $0$ in our case, and indeed we have $(\min(0,0),\min(1,0),\r)=(0,0,\r)$.
		
		Let us now see what happens inside the $1$-zone:
		\begin{mathpar}
		\inferrule*[Right=$(\lambda)$,leftskip=-6.1em,rightskip=-11.4em]{
			\inferrule*[Right=$(\lambda)$,leftskip=1.6em,rightskip=1.6em]{
				\inferrule*[Right=$(@)$,leftskip=4.5em,rightskip=4.5em]{
					y:(0,1,\tau_y)\vdash_2 y:(0,1,\tau_y)
				\and
					\inferrule*[right=$(@)$,leftskip=1em,rightskip=5.3em]{
						y:(1,1,\tau_y)\vdash_2 y:(1,1,\tau_y)
					\and
						x:(0,0,\r)\vdash_2 x:(0,0,\r)
					}{
						y:(1,1,\tau_y),\,x:(0,0,\r)\vdash_2 y\,x:(1,1,\r)
					}
				}{
					y:(0,1,\tau_y),\,y:(1,1,\tau_y),\,x:(0,0,\r)\vdash_2 y\,(y\,x):(1,2,\r)
				}
			}{
				y:(0,1,\tau_y),\,y:(1,1,\tau_y)\vdash_2\lambda x.y\,(y\,x):(1,2,\tau_y)
			}
		}{
			\vdash_2 N:(1,2,(0,1,\tau_y)\wedge(1,1,\tau_y)\arr\tau_y)
		}
		\end{mathpar}
		In the second (i.e., lower) node using the \AppR rule, the information about a positive $1$-value (coming from the left subderivation) meets the $1$-zone (coming from the right subderivation),
		and thus the $2$-value of this node is $1$.

		Denoting the type $(0,1,\tau_y)\arr\tau_y$ as $\tau^0_N$ and $(0,1,\tau_y)\wedge(1,1,\tau_y)\arr\tau_y$ as $\tau_N^1$, we continue the derivation for $M$.
		We choose to start the $2$-zone in a leaf concerning $y$.
		\begin{mathpar}
		\inferrule*[right=$(@)$]{
			\vdash_2 N:(1,2,\tau_N^1)
		\and
			y:(0,1,\tau_y)\vdash_2 y:(0,1,\tau_y)
		\and
			y:(1,1,\tau_y)\vdash_2 y:(2,1,\tau_y)
		}{
			y:(0,1,\tau_y),\,y:(1,1,\tau_y)\vdash_2 N\,y:(2,2,\tau_y)
		}
		\end{mathpar}
		This results in having a node with $3$-value $1$.
		As we want to continue in the same way with $N\,(N\,y)$ and $N\,(N\,(N\,y))$, we need to derive $(0,1,\tau_y)$ for $N\,y$ and $N\,(N\,y)$
		(which describes the situation outside of the $1$-zone):
		\begin{mathpar}
		\inferrule*[Right=$(@)$,rightskip=-5.3em]{
			\vdash_2 N:(0,1,\tau_N^0)
		\and
			\inferrule*[right=$(@)$,leftskip=1em,rightskip=5.3em]{
				\vdash_2 N:(0,1,\tau_N^0)
			\and
				y:(0,1,\tau_y)\vdash_2 y:(0,1,\tau_y)
			}{
				y:(0,1,\tau_y)\vdash_2 N\,y:(0,1,\tau_y)
			}
		}{
			y:(0,1,\tau_y)\vdash_2 N\,(N\,y):(0,1,\tau_y)
		}
		\end{mathpar}
		We continue as follows, obtaining two more nodes with $3$-value $1$:
		\begin{mathpar}
		\inferrule*[right=$(@)$]{
			\vdash_2 N:(1,2,\tau_N^1)
		\and
			y:(0,1,\tau_y)\vdash_2 N\,y:(0,1,\tau_y)
		\and
			y:(0,1,\tau_y),\,y:(1,1,\tau_y)\vdash_2 N\,y:(2,2,\tau_y)
		}{
			y:(0,1,\tau_y),\,y:(1,1,\tau_y)\vdash_2 N\,(N\,y):(2,2,\tau_y)
		}
		\and
		\inferrule*[right=$(@)$]{
			\vdash_2 N:(1,2,\tau_N^1)
		\and\hspace{-10.4pt}
			y:(0,1,\tau_y)\vdash_2 N\,(N\,y):(0,1,\tau_y)
		\and\hspace{-10.4pt}
			y:(0,1,\tau_y),\,y:(1,1,\tau_y)\vdash_2 N\,(N\,y):(2,2,\tau_y)
		}{
			y:(0,1,\tau_y),\,y:(1,1,\tau_y)\vdash_2 N\,(N\,(N\,y)):(2,2,\tau_y)
		}
		\end{mathpar}
		Next we apply the argument $a\,e$, obtaining one more node with $3$-value $1$:
		\begin{mathpar}
		\inferrule*[Right=$(@)$,leftskip=0em,rightskip=-5.5em]{
			y:(0,1,\tau_y),\,y:(1,1,\tau_y)\vdash_2 N\,(N\,(N\,y)):(2,2,\tau_y)
		\and
			\inferrule*[right=$(@)$,leftskip=1em,rightskip=5.5em]{
				\vdash_2 a:(0,1,\tau_y)
			\and
				\vdash_2 e:(0,0,\r)
			}{
				\vdash_2 a\,e:(0,1,\r)
			}
		}{
			y:(0,1,\tau_y),\,y:(1,1,\tau_y)\vdash_2 N\,(N\,(N\,y))\,(a\,e):(2,2,\tau_y)
		}
		\end{mathpar}
		In the last part of the derivation we also have a node with $3$-value $1$:
		\begin{mathpar}
		\inferrule*[right=$(@)$]{
			\inferrule*[right=$(\lambda)$,rightskip=1em]{
				y:(0,1,\tau_y),\,y:(1,1,\tau_y)\vdash_2 N\,(N\,(N\,y))\,(a\,e):(2,2,\tau_y)
			}{
				\vdash_2\lambda y.N\,(N\,(N\,y))\,(a\,e):(2,2,(0,1,\tau_y)\wedge(1,1,\tau_y)\arr\r)
			}
		\and
			\vdash_2 a:(0,1,\tau_y)
		\and
			\vdash_2 a:(1,2,\tau_y)
		}{
			\vdash_2 M:(2,2,\r)
		}
		\end{mathpar}
		
		As in Example~\ref{ex:1}, the total $3$-value of the derivation is $5$, and by adding any further $N$ to the sequence $N\,(N\,(N\,y))$, we can increase the $3$-value by $1$.
	\end{example}

	\begin{example}\label{ex:br}
		Let us also illustrate on a very simple example how the rule for the constant $b$ behaves:
		\begin{mathpar}
		\inferrule*[right=$(@)$,leftskip=-5.8em,rightskip=-7.7em]{
			\inferrule*[right=$(@)$,leftskip=5.8em,rightskip=7.7em]{
				\vdash_2 b:(0,0,(0,0,\r)\to\top\to\r)
			\and
				\vdash_2 M:(2,2,\r)
			}{
				\vdash_2 b\,M:(2,2,\top\to\r)
			}
		}{
			\vdash_2 b\,M\,e:(2,2,\r)
		}
		\end{mathpar}
		We thus simply ignore one of the arguments of $b$.
		Notice that in the second use of the application rule does not require any subderivations for the argument.
	\end{example}

%	In this rule, it is allowed (but in fact useless) that for two different $i\in I$ the type triples $(m,F_i,M_i,\tau_i)$ are equal.
%	It is also allowed that $I=\emptyset$, in which case no type needs to be derived for $Q$.

	\paragraph*{Proofs.}
	Let us now sketch the proof of Theorem~\ref{thm:nondet}.
	Condition~\hyperref[it:n1pp]{(N1$''$)} is based on the following two lemmata.

	\begin{lemma}\label{lem:decrease-m}
		Let $M$ be a closed lambda-term of sort $\otyp$ and complexity at most $m+1$.
		If $D_{m+1}$ is a derivation for $\vdash_{m+1} M:(m+1,m+1,\otyp)$, then there exists a derivation $D_m$ for $\vdash_m M:(m,m,\otyp)$
		with $\low^{m+1}(D_m)\geq\low^{m+2}(D_{m+1})$.
	\end{lemma}
	
	\begin{lemma}\label{lem:s-step}
		Let $M$ be a homogeneous and closed lambda-term of sort $\otyp$, and let $M\redb N$ be an RMF$(m+1)$ reduction.
		If $D$ is a derivation for $\vdash_m M:(m,m,\otyp)$, then there exists a derivation $E$ for $\vdash N:(m,m,\otyp)$ with the same $(m+1)$-value.
	\end{lemma}

	Condition~\hyperref[it:n2pp]{(N2$''$)} is based on two symmetric lemmata.

	\begin{lemma}\label{lem:c-step}
		Let $M$ be a homogeneous and closed lambda-term of sort $\otyp$, and let $M\redb N$ be an RMF$(m+1)$ reduction.
		If $E$ is a derivation for $\vdash_m N:(m,m,\otyp)$, then there exists a derivation $D$ for $\vdash M:(m,m,\otyp)$ with the same $(m+1)$-value.
	\end{lemma}
	
	\begin{lemma}\label{lem:increase-m}
		If $D_m$ is a derivation for $\vdash_m M:(m,m,\otyp)$, then there exists a derivation $D_{m+1}$ for $\vdash_{m+1} M:(m+1,m+1,\otyp)$
		with $\low^{m+2}(D_{m+1})\geq\log_3(\low^{m+1}(D_m))$.
	\end{lemma}
	
	Theorem~\ref{thm:nondet} is easily implied.
	Indeed, consider a homogeneous and closed lambda-term $M_m=M$ of sort $\otyp$ and complexity at most $m$, 
	its normal-form $M_0$, and lambda-terms $M_{m-1},M_{m-2},\dots,M_1$ such that all reductions between $M_i$ and $M_{i-1}$ are RMF$(i)$.

	In Condition~\hyperref[it:n1pp]{(N1$''$)} we start with a derivation $D_m$ for $\vdash_m M_m:(m,m,\r)$.
	Then, repeatedly for every $i=m-1,m-2,\dots,0$ we first apply Lemma~\ref{lem:decrease-m} to $D_{i+1}$ (with conclusion $\vdash_{i+1} M_{i+1}:(i+1,i+1,\otyp)$) 
	obtaining a derivation $D_i'$ for $\vdash_i M_{i+1}:(i,i,\otyp)$ with $(i+1)$-value not smaller than the $(i+2)$-value of $D_{i+1}$,
	and next we apply Lemma~\ref{lem:s-step} to every RMF$(i)$-reduction between $M_{i+1}$ and $M_i$,
	obtaining a derivation $D_i$ for $\vdash_i M_i:(i,i,\otyp)$ with the same $(i+1)$-value as $D_i'$.
	In effect, we obtain a derivation $D_0$ for $\vdash_0 M_0:(0,0,\r)$ with $1$-value not smaller than the $(m+1)$-value of the original derivation $D_m$.
	We conclude by observing that $D_0$ simply follows some branch of $M_0$, and that its $1$-value equals the number of constants $a$ on that branch.

	Conversely, while proving Condition~\hyperref[it:n2pp]{(N2$''$)}, at the beginning we construct a derivation $D_0$ for $\vdash_0 M_0:(0,0,\r)$, following some branch of $M_0$;
	the $1$-value of this derivation equals the number of constants $a$ on the considered branch.
	Then, repeatedly for every $i\in\{0,\dots,m-1\}$ we first apply Lemma~\ref{lem:c-step} for every RMF$(i)$-reduction between $M_{i+1}$ and $M_i$, 
	obtaining a derivation $D_i'$ for $\vdash_i M_{i+1}:(i,i,\otyp)$ with the same $(i+1)$-value as $D_i$,
	and next we apply Lemma~\ref{lem:increase-m} to $D_i'$ obtaining a derivation $D_{i+1}$ for $\vdash_{i+1} M_{i+1}:(i+1,i+1,\otyp)$ with $(i+2)$-value at most logarithmically smaller
	than the $(i+1)$-value of $D_i$.
	In effect, we obtain a derivation $D_m$ for $\vdash_m M_m:(m,m,\r)$ with $(m+1)$-value dominating the number of constants $a$ on the selected branch of the beta-normal form $M_0$.

	It remains to prove the lemmata.
	In Lemma~\ref{lem:decrease-m} we are given a derivation $D_{m+1}$ (of order $m+1$) concerning a lambda-term of complexity at most $m+1$.
	In such a derivation, a node has positive $(m+2)$-value (equal $1$) 
	if it is the closest ancestor of a node with positive $(m+1)$-value that is in the $(m+1)$-zone 
	(because all variables are of order at most $m$, the information about positive $(m+1)$-values is not passed through type environments).
	Of course every node has only one closest ancestor that is in the $(m+1)$-zone, thus the total $(m+2)$-value is not greater than the total $(m+1)$-value.
	Having this, we decrease the order of the derivation to $m$, by simply forgetting about $(m+2)$-values and about the $(m+1)$-zone; the total $(m+1)$-value remains unchanged.

	Lemmata~\ref{lem:s-step} and~\ref{lem:c-step} can be shown by performing appropriate surgeries on the derivations, like in Section~\ref{sec:deterministic}.
	One has to observe there that if a subderivation (for a lambda-term of order $m$) derives a balanced type triple, 
	then its $(m+1)$-value is $0$, and its type environment can contain only bindings with balanced type triples.
	In effect, we can treat subderivations deriving balanced and unbalanced type triples differently.
	Namely, subderivations deriving balanced type triples can be harmlessly removed or duplicated.
	Indeed, on the one hand, these operations do not change the total $(m+1)$-value.
	On the other hand, while removing such a subderivation, only bindings with balanced type triples are removed from type environments;
	this does not cause problems in nodes using the \LamR rule, because this rule allows to drop some balanced type triples.
	On the other hand, for every $k$ the surgery wants to move at most one subderivation deriving a $k$-unbalanced type triple,
	so no removal or duplication is needed for such subderivations. 
	
	In Lemma~\ref{lem:increase-m}, we have to add an $(m+1)$-zone to a derivation of order $m$.
	Starting from the root of the derivation, we repeatedly descend to the subderivation in which the total $(m+1)$-value is the greatest (arbitrarily in the case of a tie);
	the branch created in this way is taken as the $(m+1)$-zone.
	
	If, while descending from some node to its child, the total $(m+1)$-value decreases 
	(i.e., either the node itself has $(m+1)$-value $1$, or a subderivation starting in some other child also has a positive $(m+1)$-value),
	then the node gets positive $(m+2)$-value: it is the closest ancestor of some node with positive $(m+1)$-value that is in the $(m+2)$-zone.
	This can happen only in the case of the \AppR rule.
	In the \AppR rule one may observe that if a subderivation derives an $m$-balanced type triple for the argument, then its total $(m+1)$-value is $0$.
	We can thus have at most two subderivations (among those starting in children) with positive $(m+1)$-value: one for the operand, and one concerning an $m$-unbalanced type triple for the argument.
	In consequence, while descending to a subderivation, the total $(m+1)$-value decreases at most three times (with the exception that it can decrease from $1$ to $0$).
	It follows that the total $(m+2)$-value is at least logarithmic in the total $(m+1)$-value.

	\paragraph*{Extension to Recursion Schemes.} 
	
	Theorem~\ref{thm:nondet} can be also stated for infinite lambda-terms (hence, in particular, for regular infinite lambda-terms represented in a finite way by recursion schemes).
	The assumption is that we consider there only finite type derivations, and only finite branches of the generated tree (i.e., branches ending in a leaf).
	Notice that a type derivation for an infinite lambda-term can be finite, because a derivation does not need to descend to every subterm of the lambda-term.
	We claim that, under these assumptions, Theorem~\ref{thm:nondet} is true for infinite lambda-terms.

	To see this, consider a new constant $\bot$ of sort $\otyp$; it differs from $\zero$ in that we do not have a typing rule for $\bot$.
	A \emph{cut} of a lambda-term $M$ is a lambda-term obtained from $M$ by replacing some its subterms with lambda-terms of the form $\lambda x_1\lamdots\lambda x_k.\bot$ 
	(the number of the variables and their sorts are chosen so that the sort of the subterm does not change).
	It is easy to see that there is a finite derivation for $\vdash_m M:\hat\tau$ if and only if for some its finite cut $M'$ there is a derivation for $\vdash_m M':\hat\tau$, having the same values
	(we can cut off subterms not involved in the derivation).
	Likewise, the tree generated by a closed lambda-term $M$ of sort $\otyp$ contains some finite branch $B$, if and only if the tree generated by some finite cut $M'$ of $M$ contains the same branch $B$
	(the finite branch is generated after finitely many beta-reductions, concerning only a top part of $M$, and subterms located deeper in $M$ can be cut off).
	This way, the infinitary version of Theorem~\ref{thm:nondet} can be reduced to the original statement concerning finite lambda-terms.
	
	Because in a single infinite tree we can have branches with arbitrarily many constants $a$, it makes sense to give the following direct corollary of Theorem~\ref{thm:nondet}.

	\begin{corollary}\label{cor:diag}
		The following conditions are equivalent for a homogeneous and closed (potentially infinite) lambda-term $M$ of sort $\otyp$:
		\begin{compactitem}
		\item
			for every $n\in\Nat$, in the tree generated by $M$ there exists a branch with at least $n$ appearances of the constant $a$, and
		\item
			for every $n\in\Nat$, there exists a derivation for $\vdash_m M:(m,m,\r)$ with $(m+1)$-value at least $n$.
		\end{compactitem}
	\end{corollary}
	
	Because the latter condition is easily decidable for lambda-terms represented by recursion schemes, the corollary implies decidability of the former condition.

	\paragraph*{Bibliographic Note.}

	The type system presented in this section is essentially taken from Parys~\cite{itrs}; we have applied some cosmetic changes, though.

	In Parys~\cite{diagonal-types} the type system is extended to the task of counting multiple constants: 
	the $(m+1)$-value is not a number, but a tuple, where each coordinate of the tuple estimates the number of appearances of a particular constant.
	In particular, Corollary~\ref{cor:diag} is extended there to the property 
	``for every $n\in\Nat$, in the tree generated by $M$ there exists a branch with at least $n$ appearances of every constant from a set $A$'', giving its decidability.
	Deciding this property is known under the names \emph{simultaneous unboundedness problem} (SUP) and \emph{diagonal problem} (these are two different names for the same problem).

	SUP for recursion schemes was first solved in Clemente, Parys, Salvati, and Walukiewicz~\cite{downward-closure}, in a different way.
	The advantage of solving SUP using the type system presented here is twofold.
	First, the solution via the type system allows to obtain the optimal complexity, while the complexity of the original solution was much worse.
	Second, using the type system we can obtain so-called \emph{SUP reflection}: we can solve SUP simultaneously for all subtrees of the generated tree.
	More precisely, out of a recursion scheme we can create a new recursion scheme that generates a tree of the same shape as the original one, 
	but such that the label of every node contains additionally the answer to SUP in the subtree starting in that node 
	(i.e., the information whether in that node there start branches with arbitrarily many appearances of every constant from a set $A$).
	SUP reflection allowed to solve the model-checking problem for trees generated by recursion schemes against formulae of the WMSO+U logic~\cite{wmsou-schemes}.
	This logic extends WMSO (a fragment of MSO in which one can quantify only over finite sets) by the unbounding quantifier, $\unbound$. %~\cite{BojanczykU}.
	A formula using this quantifier, $\unbound X.\,\varphi$, says that $\varphi$ holds for arbitrarily large finite sets $X$.
	Let us also remark that decidability of SUP implies that given a language defined by a nondeterministic recursion scheme, it is possible to compute its downward closure~\cite{Zetzsche-down-clo},
	and given two such languages, it is possible to decide whether they can be separated by a piecewise testable language~\cite{Czerwinski-piecewise}.

	The type system presented here is also used by Asada and Kobayashi~\cite{koba-pumping-new} in their work on a pumping lemma for recursion schemes.

	The type system was inspired by the previous solution of SUP by Clemente et al.~\cite{downward-closure}.
	The idea of having balanced and unbalanced type triples, and treating them differently in type environments, comes from Asada and Kobayashi~\cite{word2tree}.

\section{Branching Quantities}

	Finally, we shortly mention one more quantity to be considered.
	In this part, suppose that the constant $a$ is of sort $\otyp\arr\otyp\arr\otyp$, that is, nodes with label $a$ have two children.
	For $n\in\Nat$, let $A_n$ be the full binary tree of height $n$, with all internal nodes labeled by $a$, and all leaves labeled by $\zero$.
	We say that $A_n$ \emph{embeds homeomorphically} in a tree $T$ if $T$ has a subtree of the form $T=a\,T_1\,T_2$ such that $A_{n-1}$ embeds in both $T_1$ and $T_2$ (defined by induction);
	$A_0=e$ embeds homeomorphically in every tree having a leaf labeled by $e$.
	Having a tree $T$ one may want to find the maximal height $n$ of a tree $A_n$ that embeds homeomorphically in $T$.
	It is an open problem how to estimate this quantity using a type system (or in any other way).
	
	\begin{openpr}\label{op:2}
		Design a type system such that the maximal value (appropriately defined) of a type derivation for a closed lambda-term $M$ of sort $\otyp$ 
		estimates the maximal number $n$ such that $A_n$ embeds homeomorphically in the beta-normal form of $M$.
	\end{openpr}

	Like in Section~\ref{sec:nondeterministic} (cf.\ Corollary~\ref{cor:diag}), existence of such a type system would solve the following problem concerning infinite lambda-terms 
	represented by recursion schemes.

	\begin{openpr}\label{op:3}
		Given a recursion scheme, decide whether for every $n$ the tree $A_n$ embeds homeomorphically in the (infinite) tree generated by the scheme.
	\end{openpr}

	A naive idea is to take the type system from Section~\ref{sec:nondeterministic}, and to change the rule for a constant $a$ into
	$\vdash_m \inc:(Z,\min(Z+1,m),(0,0,\r)\arr(0,0,\r)\arr\r)$.
	Notice, though, that if we derive a type for a tree $T$ using such a type system, 
	the value of the derivation counts the maximal number of constants $a$ in a tree that embeds homeomorphically in $T$.
	This is not what we want since, for example, if all $a$ are located on a single branch, then their number can be arbitrarily large while only $A_1$ can be embedded.
	In other words, we add values from the two children of an $a$-labeled node, while we should take their minimum.
	
	It seems that Open Problems~\ref{op:2} and~\ref{op:3} are closely related to the problem of computing the downward closures of languages of finite trees generated by nondeterministic recursion schemes
	(we remark that the downward closure of every language of finite trees is a regular language, due to the Kruskal's tree theorem).
	If we want to compute the downward closure of a language, we have to decide in particular whether it contains trees $A_n$ for all $n\in\Nat$, 
	that is, whether all $A_n$ embed homeomorphically in trees from the language.
	Like in the case of words, downward closures are also related to the problem of deciding whether two languages can be separated by a piecewise testable language.
	Goubault-Larrecq and Schmitz~\cite{schmitz-kruskal} derive a general framework for solving the piecewise testable separability for languages of trees.

	It is highly probable that Open Problem~\ref{op:3} can be solved for a subclass of recursion schemes, called safe recursion schemes, 
	using methods from Blumensath, Colcombet, Kuperberg, Parys, and Vanden Boom~\cite{quasi-weak}.
	This requires further investigation.

\bibliographystyle{eptcs}
\bibliography{bib}

\end{document}